\documentclass[hoptionsi,screen,format=acmsmall]{acmart}
\usepackage{graphicx}				
\usepackage{amssymb}
\usepackage{xypic}
\usepackage{cancel}
\usepackage{mathtools}
\usepackage{tikz-cd}

\newtheorem{theorem}{Theorem}[section]

\newtheorem{lemma}[theorem]{Lemma}

\theoremstyle{definition}
\newtheorem{definition}{Definition}[section]

\newcommand{\UBLc}{\mathcal{UBL}}

\newcommand{\Oc}{\mathcal{O}}
\newcommand{\Pc}{\mathcal{P}}
\newcommand{\Ucc}{\mathcal{U}}
\newcommand{\Jc}{\mathcal{J}}

\newcommand{\BLc}{\mathcal{BL}}

\newcommand{\band}{\mathop{\&}}
\newcommand{\bor}{\mathop{|}}
\newcommand{\Dia}{\diamondsuit}

\newcommand{\Mf}{\mathfrak{m}} 

\newcommand{\Wf}{\mathfrak{w}}
\newcommand{\Hf}{\mathfrak{h}}

\DeclarePairedDelimiter\ceil{\lceil}{\rceil}

\title[Concurrency and Dependency via Distributive Lattices]{The Topological and Logical Structure of Concurrency and Dependency via Distributive Lattices}
\author{Gershom Bazerman}
 \affiliation{%
   \institution{Awake Security}}
 \email{gershomb@gmail.com}
\author{Raymond Puzio}
\affiliation{%
   \institution{Albert Einstein Institute}}
\email{rspuzio@gmail.com}

\begin{abstract}
This paper is motivated by the desire to study package management using the toolkit of the semantics of functional languages. As it transpires, this is deeply related to the semantics of concurrent computation. The models we produce are not solely of theoretical interest, but amenable to analysis and computation. This work makes a number of related contributions. First, it relates the specification of branching dependency structures, which exist in fields from knowledge-representation to package management, to the specification of semantics of concurrent computation. Second, it relates dependency structures to lattices in a precise way, establishing a full correspondence with a particular subclass of lattices. It then makes use of this as a key ingredient, coupled with the underappreciated Bruns-Lakser completion, in relating dependency structures to locales -- objects equipped with both topological and logical properties. It then provides an example of how this interplay of properties can be of use -- using topological properties of the dependency structure to equip internal logics of associated locales with a modality representing contraction relations (i.e. ``versioning''). This approach lets us see linking (or rather, the choice of what to link against, i.e. ``solving'') as an effect. Finally, it discusses how such constructions may relate to important questions in complexity theory, including solutions of satisfiability problems. Along the way, we will see how this approach relates to familiar objects such as package version policies, Merkle-trees, the nix operating system, and distributed version control tooling like git.
\end{abstract}

\begin{document}
\maketitle

\section{Introduction}
This project began with seeking to understand the mathematical structure and logic of software package repositories. Such repositories contain tens of thousands of packages, with complex webs of interlinking dependencies, represented as expressions in propositional logic, containing not only branching choices, but also a notion of ``compatibility ranges'' and a notion of conflict. Rather than develop a language and then study possible semantics for it, we begin from the ``ground up'' by seeking to first model dependency structures in a very general way, and then tease out the logical structure already latent within this setting.

The structures we developed for representing dependencies turned out to be extremely similar to work on the semantics of concurrent computation -- and for good reason! The problem of branching dependency specification is the same as the problem of concurrent computation, just ``turned on its head''. Intuitively, a package repository (such as Hackage, or npm, or even as provided by Debian) may be seen as a specification of a concurrent program, which is ``executed'' by a user, at each step, just picking one more thing to install from the collection of things whose dependencies are already installed. Conversely, it is apparent how one might, given a concurrent program, generate a package repository which corresponds to it.

Semantics of concurrent computation, at base, consist of a collection of states, and certain allowable transitions between them, which may be simultaneous, and which may be nondeterministic. A dependency specification, such as given by a package repository, also has a collection of states (the collection of installed packages), and also has a collection of allowable transitions (one may only install a new package when all dependencies are satisfied). Furthermore, concurrency takes the form of distinguishing when two independent packages may be built simultaneously. And finally, in both cases, we have a notion of ``incompatibility'' -- the former, in terms of contention for a shared resource, and the latter, in terms of e.g. disallowing that two linked packages expose the same required symbol with different definitions.

The difference is then largely in the questions asked about such structures. With concurrent semantics, the whole structure is the ``program'' and the typical questions asked are how such things compose. With dependency structures, a ``program'' is what we think of as a ``build plan'' -- a single trace through the structure to a particular end state, and the questions asked are about optimality, reachability, etc. Our approach is inspired by the latter way of thinking, but we think it sheds light on many related issues as well. In particular, we see that ``inside'' any single concurrent program, we can examine not only its state space, but also a related internal logic -- the logic of dependency specifications. This lends itself to fine-grained intensional analysis, exploiting the interplay. In the course of this, we also see how ``linking'' may be seen as an effect. When a user of Debian runs the command ``apt-get emacs'', this may bring in one of a number of versions of emacs, or perhaps the same version of emacs linked against one or another version of glibc. ``Solutions'' to commands are not unique, and may vary due to the state of installed packages, the state of the upstream repositories, algorithms chosen by a particular version of a solver, etc. Our approach allows us to view the specialization of a command to a particular solution as an effect.

A third way to see DSCs is as models of knowledge representation. Instead of programs or packages one installs, we can imagine studying dependencies among mathematical facts. Certain definitions and theorems are necessary to understand others, and so forth. Again, we have collections of states (things which can be known), allowable transitions (things which provide the basis for learning other things), and potential conflicts (things which you cannot simultaneously believe -- e.g., at least according to some, one cannot be a baysian and a frequentist at the same time, although you can trade off on alternate weekdays).

In section 2, we introduce the basic elements and tools -- dependency structures with choice, and their related ``trace'' structures, reachable dependency posets. Section 3 presents a number of results and constructions in order-theory relevant to reachable dependency posets. First among them is the Bruns-Lakser completion, for which we provide a simple calculation in the finite case that yields a general formal mathematical notion of what a ``Merkle'' structure is. Section 3 also introduces a novel representation theorem, showing that dependency structures with choice correspond to a very particular subclass of lattices -- those which either distributive, or upper semimodular but not modular. It concludes by extending this construction to a class of locales. Section 4 discusses how ``versioning'' information can be thought of as ``covering'' relations in dependency structures, and can give rise to a topological operation on their resultant locales. Section 5 then gives  two intuitionistic logics that can be built over the resultant locales  -- one an immediate ``logic of paths'', and one passing through the free distributive lattice over a poset to yield a ``logic of requirements''. It also discusses how the topological operator interprets in these logics as a modality. Section 6 presents some preliminary investigations on how this approach may shed light on the innate topological structure of dependency problems, with regards to the difficulty they present to solvers. We then conclude with a discussion on related and future work.

\section{Dependency Structures with Choice}
We begin our analysis with an algebraic definition of dependency structures. These are intended to correspond almost immediately to the data provided by package repositories -- events (packages) which may depend on a choice of other events. These structures do not (yet) have any notion of a choice of versions -- so version 1.0 and 1.1 of the same package are logically two entirely different events. Later we will see how to recover this data.

\begin{definition}
A \textbf{Pre-Dependency Structure with Choice} is a pair \((E, D : E \rightarrow \Pc(\Pc(E)))\) where \(E\) is a finite set of events, and \(D\) is a non-nullary mapping from \(E\) to its double powerset, to be interpreted as mapping each event to a set of alternative dependency requirements -- i.e. to a predicate in disjunctive normal form ranging over variables drawn from \(E\).
\end{definition}

\begin{definition}
A \textbf{Dependency Structure with Choice} (DSC) is a pre-DSC with \(D\) satisfying  appropriate conditions of transitive closure and cycle-freeness. We define \(X\) as a \textbf{possible dependency set} of \(e\) if \(X \in D(e)\). We call an event set \(X\) a \textbf{complete event set} if for every element \(e\) there is a possible dependency set \(Y\) of \(e\) such that \(Y \subseteq X\). A pre-DSC is a DSC if every possible dependency set of every element is complete, and no possible dependency set of any element contains the element itself. Pre-DSCs may be completed into DSCs by repeatedly taking transitive completion of possible dependency sets (with regards to each transitive possible dependency set) and then deleting cyclic sets (and elements whose possible dependency set becomes empty under this process) until a fixpoint is reached.
\end{definition}

DSCs are richer than the standard notion of a dependency tree or dependency graph. In such structures, a node \(a\) with edges to \(b\) and \(c\) exhibits a dependency on both \(b\) and \(c\). There is no way, however, to express a dependency on either \(b\) or \(c\). A domain-theoretic account of such structures is given by \textit{pomsets}, introduced by Vaughn Pratt \cite{pratt1986modeling}. Pomsets are a special instance of a broader class of structures, known as \textit{event structures}, introduced by Nielson, Plotkin and Winskel \cite{nielsen1981petri}, and used in the domain-theoretic semantics of concurrent computation and concurrent games. Such structures have not only a (choice-free, transitively normalized) dependency relation, but in addition a conflict structure which indicates incompatible collections of events (typically presented as a collection of consistent sets which carves out only compatible collections of events). Finally, there are so-called \textit{general event structures}, which extend event structures with a notion of choice in roughly the same fashion as DSCs (i.e. by moving from a partial ordering relation into a relation between elements and powersets). Their theory is less well behaved and understood, and its study is an area of ongoing work. DSCs may be characterized as general event structures which are conflict-free. Thus intuitively (i.e. not necessarily formally), in a lattice of expressive power, DSCs sit above pomsets, ``side by side'' with event structures, and below general event structures. One hope of the present is work is that it might be usefully extended in some fashion with conflicts, leading among other things to a further understanding of general event structures.\footnote{When first introduced, what we now know as ``general event structures'' were simply named ``event structures'', and what we now know as ``event structures'' were named ``stable event structures''. We follow the modern convention in this paper.}

\subsection{Reachable Dependencies Posets of a DSC}

Data as given in a DSC is purely declarative. To introduce an analysis of dynamics, we need a structure which we can trace through time. From DSCs we derive a partially ordered set of execution traces, analogous to the family of configurations of an event structure. Order-theoretic terminology used here will be reviewed in the next section.

\begin{definition} The \textbf{reachable dependency poset} of a DSC is the result of an operation, \(rdp\), which sends DSCs to bounded posets (i.e. posets with top and bottom elements) by the following two-step procedure:  We take as elements all collections of events, i.e. \(\Pc(E)\), and impose the least order relation such that one collection of events, \(X\), is above another, \(Y\), if \(Y \subset X\) and for every element of \(X\), there is a possible dependency set contained in \(Y\).  The conditions of transitivity and cycle-freeness ensure that under such an ordering, every event will lie above the empty event set. We refer to event sets that exist in this poset as \textbf{reachable event sets}.
\end{definition}

Informally, a reachable dependency set is generated by asking ``for each event, what are the basic (reachable) event sets which contain it',' and then completing those by the empty set and all unions of this basis. This is sometimes known as ``unwinding''. Viewed as a graph, nodes of a reachable dependency poset correspond to complete event sets, and edges correspond to linear accretion of event sets over time by addition of subsequent events. Consequently, a reachable dependency poset may also be seen as generated by considering all possible dependency sets of all events, augmenting each with the event itself, and then, under the inclusion ordering, augmenting the result with the empty set and in addition all possible joins. We leave it to the reader to convince themselves that both procedures yield the same result.

It follows that for any DSC \((E,D)\), \(rdp((E,D))\)  is a subposet of \(\Pc(E)\), and which has all joins as unions. Further, as a bounded poset with all joins, by the adjoint functor theorem for posets, it also has all  meets, and is hence a lattice. However, importantly, meets do not correspond to intersections. Consider a DSC containing an event \(a\), which depends on either \(b\) or \(c\). \(\{a,b\}\) and  \(\{a,c\}\) are reachable event sets, but their intersection, \(\{a\}\), is not. The meet of two reachable event sets is the greatest reachable event set that is a subset of their intersection. This necessarily exists, and is generated by taking the union of all reachable subsets of the intersection. Note that this reduction operation preserves joins of reachable event sets, but not necessarily joins of \textit{all} event sets. As a syntactic convenience, when we denote operations whose domain is a lattice on a DSC, we implicitly pass through the reachable dependency poset construction.

If we consider DSCs as an algebraic ``signature'', then RDPs provide the models of this signature. As we shall see, the connection between this algebraic structure and its order-theoretic models can be made precise through a representation theorem.

\section{Distributive Lattices and the Idempotent Distributive Lattice Completion}

We review here some basic facts and notation regarding order theory and lattices.

A \textbf{partially ordered set} or poset, \(P\) is a set equipped with a partial order relation \(\le\), which is transitive,  reflexive, and antisymmetric (i.e. for which \( a \le b \band b \le a \implies a = b\)). A (homo)morphism of posets is an monotone (order-preserving) function on their elements, and with such morphisms posets form the category \(Pos\) and finite posets form the subcategory \(FinPos\). Two posets are equivalent when there exist morphisms \(f, g\) between them such that \(f \odot g = id\) and \(g \odot f = id\), i.e. when they are equivalent as objects of \(Pos\). We note that all posets have a standard partial order on them such that \(P \le Q\) when there exists an order-preserving embedding \(P \rightarrow Q\).

A \textbf{lattice}, \(L\) is partially ordered set for which every two elements have a unique greatest lower bound, their \textbf{meet} (\(\wedge\)) and a unique least upper bound, their \textbf{join} (\(\vee\)). The join and meet operations of a lattice are necessarily commutative, associative, and idempotent. A (homo)morphism of distributive lattices is a morphism of posets which also preserves meets and joins. A \textbf{join-semilattice} and \textbf{meet-semilattice} are posets that respectively have all finite joins or all finite meets. A \textbf{complete lattice} is a lattice which has joins and meets of infinitary as well as finitary collections of elements. We write \(\bigwedge\) and \(\bigvee\) for the meet and join operations as applied to an entire set of elements. By abuse of notation, we also may write, e.g., \(x \vee S\) where \(x\) is an element of a poset and \(S\) is a set of elements, to indicate the lifting of application of the unary operation \(x \vee -\) to every element in the set.

A \textbf{distributive lattice}, is a lattice satisfying the additional property that for all \(x, y, z\) in \(L\), \(x \vee (y \wedge z) = (x \vee y) \wedge (x \vee z)\). It is easy to verify that if this condition (join distributing over meet) is satisfied, then the dual condition (meet distributing over join) is also satisfied.  Lattice homomorphisms between distributive lattices are necessarily distributive lattice homomorphisms, and with such morphisms distributive lattices form the category \(DLat\) and finite distributive lattices form the subcategory \(FinDLat\). In \(FinDLat\), all lattices necessarily have a unique top and bottom element (i.e. are bounded). As such, we require morphisms in \(FinDLat\) to also preserve top and bottom elements as the nullary join and meet (i.e. to be homomorphisms of bounded lattices).

A \textbf{join-irreducible} element of a poset is an element \(x\) such that no collection of elements not including \(x\) has \(x\) as its join. The operation \(\Jc(P)\) sends a poset (or a lattice viewed as a poset) to the sub-poset of its join-irreducible elements, sharing the same order relation. An intuition that this lends itself to is that join-irreducible elements are ideals. We refer to elements of a poset which are not join-irreducible as \textbf{composite} elements, and the set of join-irreducible elements which joins to them as their \textbf{basis}. It is important to note that if a poset has a globally least element (i.e. element which stands below all other elements in the order relation), that element is not join-irreducible, since it is the join of the empty set. However, if a poset has more than one locally least element (i.e. element with no element below it), then all such elements are join-irreducible. It is also important to note that even if an element is join-irreducible in \(P\), it still may nonetheless become a join in the restriction to \(\Jc(P)\).

A \textbf{downset} of a poset is a set of elements of the poset which is downwardly-closed -- i.e. for which \(x \in S \band y \le x \implies y \in S\). The operation \(\Oc(P)\) sends a poset to the poset of its downsets, ordered by inclusion. Such a poset has meets and joins as respectively intersection and union, and consequently is a distributive lattice. Further, \(\Oc(P)\) is a morphism (and in fact an embedding) of posets, which sends each \(x \in P\) to the set \(\{y \mathbin{|} y \le x\}\). The dual operation to taking downsets is taking \textbf{up-sets}  which are upwardly-closed. We denote this as \(\Ucc(P)\).

A \textbf{Heyting algebra} is a lattice with a unique top and bottom element, and a special ``implication'' operation called the \textbf{relative pseudo-complement} (\(a \rightarrow b\)) which yields the unique greatest element \(x\) such that \(a \wedge x \le b\). A \textbf{complete Heyting algebra} is a Heyting algebra such that it is also a complete lattice. The category of complete Heyting algebras takes as morphisms monotone functions which preserve finite meets, arbitrary joins, and implication.

A \textbf{frame} is a complete Heyting algebra. However, the category \(Frm\) of frames takes as morphisms monotone functions which preserve finite meets and arbitrary joins, but not necessarily implication. This is to say that the relative pseudo-complement operation derived from finite meets and arbitrary joins necessarily exists in frames, but may not commute with any given frame homomorphism. In the finitary case, distributive lattices and complete Heyting algebras coincide, and hence \(FinFrm = FinDLat\).

A \textbf{locale} is again the same thing as a frame. However, in the category \(Loc\) of locales, morphisms are viewed reversed, and hence \(Loc = Frm^{op}\) and \(FinLoc = FinFrm^{op} = FinDLat^{op}\)

In passing, we will make use of Birkhoff duality, which we recall here as well.

\textbf{Thm. (Birkhoff duality)}: When L is a finite distributive lattice, \(\Oc(\Jc(L)))\) is an equivalence, and for any finite poset P,  \(\Jc(\Oc(P)))\) is an equivalence. Further, this equivalence extends to a functorial equivalence between the categories \(FinPos\) and \(FinDLat\), with monotone functions on posets corresponding to homomorphisms of distributive lattices. As a consequence, all finite locales may be viewed in terms of their underlying posets.

\subsection{Distributive Lattices in Logic and Topology}

Distributive lattices play a very special role in both logic and topology. From a logical standpoint, Heyting algebras provide a complete semantics for intuitionistic logic, with true and false corresponding to top and bottom, and corresponding to meets, or corresponding to joins, and implication being given by the relative pseudo-complement. More precisely, a formula in intuitionistic propositional logic is provable (tautological) if and only if it is valid (yields true under every assignment of variables) in every Heyting algebra (and in fact this result holds even when one considers only finite Heyting algebras). Hence, a finite Heyting algebra (resp. distributive lattice) is a setting where we can directly interpret expressions in intuitionistic logic \cite{van1988troelstra}.

From a topological standpoint, frames are the object of study of locale theory, i.e. so-called ``pointless topology''. From any topological space -- given as a set of points, and a covering relation of open sets -- the open sets themselves form an order theoretic structure which is precisely a frame (which is known, when the homomorphisms between frames are viewed backwards, as a ``locale''). If we then forget the points, and consider only the frame, we can still ``do topology'' -- and from any frame (resp. locale) we can recover a special type of space, known as a sober space. In fact, frames and sober spaces are in one-to-one correspondence \cite{johnstone1982stone, vickers1996topology}.

We will demonstrate that reachable dependency posets correspond one-to-one with a broad class of lattices. This class clearly is broader than distributive lattices. Thus it does not provide a setting in which we can perform logical or topological analysis directly. However, intuitively, it feels like it \textit{should} give such a setting. In particular, events resemble points, and reachable event sets very much resemble open covers. This motivates the construction of the ``best'' way to derive an associated distributive lattice from a reachable dependency poset, and a further correspondence theorem that lets us establish an equivalence between DSCs and a special class of locales.

\subsection{Bruns-Lakser Completion}

In 1970, Bruns and Lakser introduced an indempotent distributive lattice completion for meet-semilattices \cite{bruns1970injective}. Only muchl later \cite{ball2016dedekind}, was it realized that this completion (though expressed in different terms) was actually introduced by  Holbrook MacNeille in the same 1937 work where he first introduced the famed Dedekind-MacNeille completion of partially ordered sets into complete lattices \cite{macneille1937partially}. 

First, we recall the original construction of Bruns-Lakser and MacNeille.

\begin{definition}
(Bruns-Lakser) An \textbf{admissible set} is a subset \(S\) of a meet-semilattice \(P\) in which for all \(x\), \(x \wedge \bigvee(S) = \bigvee(x \wedge S)\).
\end{definition}

\begin{theorem}
(Bruns-Lakser, MacNeille) The partially ordered set of all admissible sets of a meet-semilattice \(P\) is a distributive lattice, \(\BLc(P)\). There exists an injection \(bl : P \rightarrow \BLc(P)\), which preserves all meets and joins of admissible sets. Furthermore, any morphism to a distributive lattice that preserves all meets and joins of admissible sets, \(f : P \rightarrow D\), factors uniquely into the injection \(bl : P \rightarrow \BLc(P)\) followed by a distributive lattice homomorphism \(g : \BLc(P) \rightarrow D\), as in the following diagram:
\begin{equation*}
\begin{tikzcd}
\BLc(P) \arrow[r, "!",dashed]            & D \\
P \arrow[u, hook] \arrow[ru, "f"'] &
\end{tikzcd}
\end{equation*}

\end{theorem}

The completion is idempotent. Further, it is functorial. \(\BLc\) acts on morphisms that preserve meets and admissible joins (i.e. joins of admissible sets) by lifting their action on all elements in the original source and target, and extending their action on new elements in the source by sending them to the join in the target of the targets of their join-irreducible basis in the source. This functor gives the category of distributive lattices as a reflective category of the subcategory of meet-semilattices which shares all objects but only has morphisms that preserve meets and admissible joins \cite{gehrke2014distributive}. This is to say, given \(i\) as inclusion:

\begin{equation*}
  (\BLc \dashv i) :  DLat \stackrel{\stackrel{\BLc}{\leftarrow}}{\hookrightarrow}
  DSLat \subset SLat
  \,.
\end{equation*}

The collection of admissible sets is rather large and unwieldy. But in the finite case, we have a much nicer characterization of the completion, which is computationally simple and also suggestive and familiar with regards to structures that occur elsewhere in computer science. The following is a simplification and extension of MacNeille's characterization of the finite elements of his completion, which holds even for posets which are not meet-semilattices:

\begin{lemma}
For a finite poset, \(\BLc(P)\) may be constructed as \(\Oc(\Jc(P)))\), with an injection that sends join-irreducible elements to their downsets, and composite elements to the union of their join-irreducible basis.
\end{lemma}

\subsection{Merkle Structures and the Completion of Reachable Dependency Posets}

\begin{figure}
\begin{minipage}[c]{0.3\textwidth}
\begin{equation*}
    \xymatrix{& abc \ar@{-}[dl] \ar@{-}[d] \ar@{-}[dr] & \\
      \underline{ab} \ar@{-}[d] \ar@{--}[dr] & bc \ar@{-}[dl] \ar@{-}[dr] &
        \underline{ac} \ar@{-}[d] \ar@{--}[dl] \\
        \underline{b} & \xcancel{a}  & \underline{c} }         
\end{equation*}
\end{minipage}
\begin{minipage}[c]{0.08\textwidth}
  \begin{equation*}
    \xymatrix{\ar@2{~>}[r]^{\BLc} &}
  \end{equation*}
\end{minipage}
\begin{minipage}[c]{0.3\textwidth}
\begin{equation*}
    \xymatrix{& a_ba_cbc \ar@{-}[dl]  \ar@{-}[dr] & \\
      \mathbf{a_bbc} \ar@{-}[d] \ar@{-}[dr] & &
        \mathbf{a_cbc} \ar@{-}[d] \ar@{-}[dl]\\
        \underline{a_bb} \ar@{-}[d] & bc \ar@{-}[dl] \ar@{-}[dr] &
          \underline{a_cc} \ar@{-}[d] \\
         \underline{b} \ar@{-}[dr] & & \underline{c} \ar@{-}[dl] \\
      & \mathbf{\emptyset} }
\end{equation*}
\end{minipage}
\caption{A simple example of the extended Bruns-Lakser completion, join irreducible elements (and their image) underlined, new elements in bold}
\label{Fig1}
\end{figure}
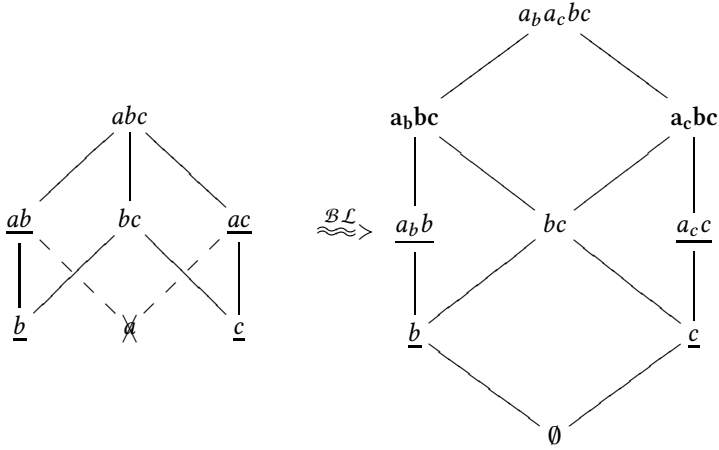


To the vaguely formed question ``how do we make reachable dependency posets topological'' we can now propose a precise answer: application of the idempotent distributive lattice completion. This is to say, we take the composite \(\BLc(rdp(E,D))\) (which, by notational shorthand, may be written as \(\BLc(E,D)\)). The irreducible elements of a reachable dependency poset are those sets which are generated by the possible dependency sets of individual events -- i.e. that have an event which is shared by no complete dependency set below them. In the example where an event (or package, if you prefer) \(a\) depends on either \(b\) or \(c\), the irreducible elements, i.e. \(\Jc(rdp(E,D))\), are \(\{b\}\), \(\{c\}\), \(\{a,b\}\), and \(\{a,c\}\). Hence the application of \(\Oc\) yields the four sets \(\{\{b\}\}\), \(\{\{c\}\}\), \(\{\{b\},\{a,b\}\}\), and \(\{\{c\},\{a,c\}\}\), but also \(\{\}\), \(\{\{b\},\{c\},\{b,c\}\}\) and \(\{\{b\},\{c\},\{a,b\},\{a,c\}\}\) (see Figure \ref{Fig1}).

In the case of finite posets presented as a collection of sets, with ordering given by inclusion (as in the case of reachable dependency posets), we can give a simple algorithmic characterization of this construction. For every event which has multiple ``paths'' to enable it (i.e. multiple possible dependency sets), split the event into new events, each labeled by a different possible dependency set. And since branched events may depend on other branched events, do so recursively. In the resulting structure, rather than sets of events, there are sets of events each labeled by the ``path'' taken to get to them. 

From the standpoint of concurrent semantics, this amounts to replacing sets of events by sets of execution traces. From the standpoint of dependency specification, this amounts to augmenting packages by their ``build plan''. This operation makes intuitive sense in that it provides a more granular and correct specification of what a package ``really'' means. For example library \(a\) may link against library \(b\) or \(c\), each of which provide the same API-surface, but which have subtle differences in behavior. So while \(a\) depends ``equally'' on either \(b\) or \(c\), the resulting products, \(a_b\) and \(a_c\) are not guaranteed to be the same thing. It is precisely this distinction which is captured by taking the idempotent distributive lattice completion.

The subscript ``shorthand'' used above renders each event as subscripted by the path of dependencies to reach it. Since that path itself consists of events, in a nontrivial chain, then those events too are subscripted, and soforth. Representing this computationally seems a bit of a chore. If we took each label and turned it into a hash, and then when taking sets of labels instead took hashes of their hashes, etc, then (with high probability) we could represent the same information in constant space rather than space geometric in the height of our poset.

This operation -- augmenting a structure such that each element contains a hashed description of the path to reach it -- is relatively ubiquitous in modern software. It lies at the heart of the model of patches in distributed version control systems such as git, and is also used in blockchains. In those cases, while meets exist, joins do not. And further, in the blockchain case, it is even simpler because the goal is to restrict ``truth'' to a chosen linear path -- i.e. to avoid branching or ``splits'' in the chain. This construction is also the basic insight of the nix operating-system, as well as the ``nix-like'' store used by the Cabal build system for Haskell. In these cases, meets and joins both are used. The operation of generating the path information and taking a hash-chain of it is known as producing a \textbf{Merkle tree}\cite{merkle1987digital}. The idempotent distributive lattice completion, \(\BLc\) is then, in a sense,  the \textit{non-probabilistic Merkle transformation} of a poset, and provides a formal description of what it means when we take an existing data structure and ``turn it into a Merkle tree''. \footnote{Note that it is non-probabilistic because the completion does not describe the hashing component of a Merkle tree, which is probablistic, and instead captures only the uniform algebraic structure in the case where hashing is unique.} 

We believe this mathematical characterization helps explain the success of the nix paradigm. Through associating packages with their full dependency trace, users (ideally) no longer have to reason about complex dependency chains, but can instead use purely set-theoretic reasoning, freely taking intersections and unions of desired ``atomized'' target packages. There do remain obstructions to ``living the dream'', in the form of incompatibilities, which we hope to address more thoroughly in future work.

It would be remiss of us not to mention, as well, that the operation of sending elements of a poset to their downsets, a key step in this transformation, is a special case of the Yoneda embedding. In particular, it is the Yoneda embedding for (0,1)-presheaves -- i.e. when \(C^{op} \to Set\) is decategorified to \(P^{op} \to Bool\). Categorically-inclined readers may bear this in mind the next time somebody asks them what practical purpose such constructions serve.

\subsection{The General Relationship between DSCs and Lattices}
Given that every DSC induces a finite lattice, it is reasonable to ask the general relation between finite lattices and DSCs. 

First, we consider the lattices induced by dependency structures with no choice -- i.e. where every event has a single possible dependency set. In such a case, the dependency structure presents a poset (and indeed, up to renaming, every poset induces a unique dependency structure). The reachable dependency poset of such a structure is generated by considering down-closed sets of events -- i.e. those where each event can occur only if all its dependencies occur. By Birkhoff duality, this is a distributive lattice, and thus choice-free dependency structures are equivalent to distributive lattices.

When choice appears, we must therefore land in non-distributive lattices. Birkhoff showed that there are two canonical ways in which a lattice can fail to be distributive -- through containing as a sublattice the ``forbidden configurations'' \(M_3\) or \(N_5\). Possibly containing \(N_5\) but not containing \(M_3\) corresponds to the weaker logical property of modularity. As we shall see below, it is impossible for the \(rdp\) construction to generate \(M_3\) as a sublattice. However, in the running example we have used of a nontrivial DSC (where \(a\) depends on \(b\) or \(c\)), the lattice \(S_7\) is generated (Figure \ref{Fig3}). This lattice contains two copies of \(N_5\) as a sublattice (the one generated by excluding \(bc\) and \(c\) and the one generated by excluding \(bc\) and \(b\)). Furthermore, it is the canonical example of a lattice which fails to be modular, but nonetheless satisfies the still weaker upper semimodularity condition. In fact, a lattice that is upper semimodular but not modular must contain as a cover-preserving sublattice \(S_7\)  \cite{stern1999semimodular}. 

\begin{figure}
\begin{minipage}[c]{0.3\textwidth}
\begin{equation*}
    \xymatrix@=1em{& \top \ar@{-}[dl] \ar@{-}[d] \ar@{-}[dr] & \\
      a \ar@{-}[dr] & b \ar@{-}[d] & c \ar@{-}[dl]  \\
       & \bot &}
\end{equation*}
\end{minipage}
\begin{minipage}[c]{0.3\textwidth}
\begin{equation*}
    \xymatrix@=1em{& \top \ar@{-}[ddl] \ar@{-}[dr] & \\
      & & b'  \ar@{-}[d]  \\
      a \ar@{-}[dr] & & b \ar@{-}[dl] \\
      & \bot &}
\end{equation*}
\end{minipage}

\caption{The non-distributive lattice \(M_3\), aka the ``chinese lantern'', and the non-distributive, non-modular lattice \(N_5\), the pentagon.}
\label{Fig2}
\end{figure}
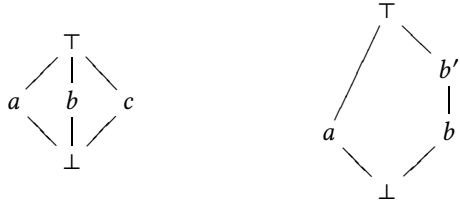

\begin{figure}
\begin{minipage}[c]{0.3\textwidth}
\begin{equation*}
    \xymatrix{& \top \ar@{-}[dl] \ar@{-}[d] \ar@{-}[dr] & \\
      ab \ar@{-}[d]& bc \ar@{-}[dl] \ar@{-}[dr] &
        ac \ar@{-}[d]  \\
        b \ar@{-}[dr]  & & c \ar@{-}[dl] \\
        & \bot &
        } 
\end{equation*}
\end{minipage}
\caption{The lattice generated by the unwinding of a simple dependency structure with choice, where \(a\) may depend on \(b\) or \(c\). This is an instance of the lattice \(S_7\), the ``centered hexagon''.}
\label{Fig3}
\end{figure}
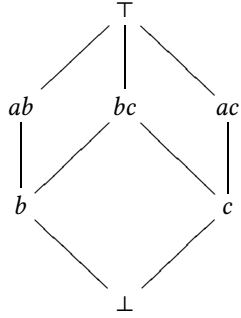

It turns out that when the lattices generated by \(rdp\) fail to be distributive, they do so precisely by containing copies of \(S_7\) -- i.e. they are upper semimodular, but not modular. That is to say that the generated lattices fail to be distributive \textit{only} by also failing to be modular, but do not fail so badly that they cease to be upper semimodular. Not only does \(rdp\) land precisely in this class of lattices, but it lands so precisely that we can construct an inverse, and a full equivalence.

\begin{theorem}
\label{representation}
There exists a map \textbf{uds} (underlying dependency structure) from upper semimodular lattices without \(M_3\) as a sublattice to DSCs such that  \(rdp \circ uds\) is the identity, and \(uds \circ rdp\) is an equivalence up to renaming.
\end{theorem}

\begin{proof}
First, we show that \(rdp\) only generates upper semimodular lattices without \(M_3\) as a sublattice. Next, we construct an inverse, and show that on such lattices, it is the identity.

We know the generated lattice is upper semimodular by following Birkhoff's condition for weak semimodularity (which corresponds to semimodularity in the finite case). An element is defined to cover another element if it is greater than that element, and there is no intervening element between them. A finite lattice is semimodular iff for each pair of elements which are both covered by their join, they both cover their meet. In the case of the \(rdp\) construction, an event set covers another event set precisely when it differs by adding a single event. So for two elements to be covered by their join, they must each share all events save one. Necessarily, then, they will cover their meet, which is their intersection.

Next, we demonstrate that \(rdp\) must generate lattices that are free of \(M_3\). If a lattice has \(M_3\) as a sublattice, then it must contain three elements that are relatively unordered and for which any two elements will join to the join of all three. Further, for these three elements, their meets must fulfill the dual condition. Now, assume three elements have a common meet. Then, they must each differ by distinct events.  Their joins are freely generated by union, and thus they must have three different joins, and therefore the desired construction is impossible.

Now, we construct \(uds(L)\) in two steps, first establishing the ``free'' DSC \(FD\) on the lattice \(L\), and then identifying events which are the ``same''. \(FD(L)\) has an event set consisting of the join-irreducible elements of \(L\), and as dependencies singleton sets freely generated by the ordering relation on the elements.  I.e. \(FD(L) = (\Jc(L), x \mapsto \{\{ y \in \Jc(L) \mid y<x\}\})\). To identify two events is to substitute all references for first to the second, and further, to define the new collection of possible dependency sets of the first to be the union of the existing collections given by both events. We denote a DSC \(D\) with two events, \(a\) and \(b\) identified as  \(D/[a\sim b]\). The set \(Q(L)\) consists of pairs of join-irreducible elements \((a,b)\) such that  \(\exists\, x \in L . \,  x < a \vee b \band x \vee a = x \vee b = a \vee b\). (Note here that \(x\) need not be join-irreducible). Now, \(uds(L) = FD(L)/[a \sim b,a' \sim b',...]\) for all \((a,b), (a',b'), ...\in Q(L)\).

Finally, we show that \(rdp \circ uds\) is the identity. When the source lattice is distributive, then the quotient operation trivializes (in a distributive lattice, joins are necessarily unique, and hence \(Q(L)\) is empty). Therefore, by Birkhoff duality, we have the identity. When the source lattice fails to be distributive, we know that it does not have \(M_3\) as a sublattice, and it is upper semimodular. As discussed above, a lattice that is upper semimodular but not modular must contain as a cover-preserving sublattice \(S_7\), precisely our running example. An instance of \(S_7\) then results in an element of Q(L). For example, in Fig. \ref{Fig3},  we have the pair \((ab,ac)\), and there exists an \(x\) given by \(bc\) satisfying the desired condition (i.e. it is less than the join of the pair, and joining it with any element of the pair is equal to the join of the pair). Given a pair in \(Q(L)\), considered as two event sets, the existence of a third ``unifying'' element (necessarily unordered with regards to them) could come about in two ways. Firstly, it could be a join irreducible element. But that would induce as a sublattice \(M_3\), because all three elements would share the same join, and furthermore they must share the the same meet (otherwise that element would be reducible as the join of the distinct meets). Therefore the element is necessarily join reducible, and since it must be unordered with regards to the pair, it must be the join of two elements below the pair. Now, in the \(rdp \circ uds\) roundtrip, this ``unifying'' element will be discarded by \(uds\) because it is join reducible, and reconstructed by \(rdp\). Furthermore, the two ``unified'' elements will be quotiented into a single event with two possible dependency sets in one direction, and then ``exploded'' back into two distinct elements in the return trip. The rdp construction guarantees their structure below is preserved by such a roundtrip, and the fact that the two events were identified ensures that no additional structure above was created. This completes the proof.
\end{proof}

\subsection{DSCs and Finite Locales}

We have seen how a DSC may be transformed into a finite locale via the extended Bruns-Lakser completion of its reachable dependency poset. It is natural to ask how far this relationship extends. It is almost immediate that any finite locale may be generated by a DSC. All that is necessary is to consider the join-irreducible elements of the frame as events, each of whose single possible dependency set is all other events below it in the frame ordering. However, even up to renaming of events, distinct DSCs can nonetheless present order-equivalent frames. For example, the DSC with \(a\) depending on \(b\) or \(c\) gives the same frame as the DSC with \(a\) depending on \(b\) and \(d\) depending on \(c\) (since the Bruns-Lakser completion ``splits'' the former \(a\) into two copies to begin with). 

This tension (multiple presentations of the same structure) occurs very frequently in topos theory (and its decategorification in locale theory) when multiple sites (categories with a topology, resp. posets with a coverage) can present the same topos (resp. locale). This suggests that a tighter correspondence should not be to locales directly, but rather to their presentations as sites (known, for locales, as ``posites''). A good overview of posites is provided by \cite{schultz2017temporal}.

The central tool in this section is the nucleus, which is the localic analogue of a topology. We recall its definition and some basic facts regarding it. Readers are referred to \cite{johnstone1982stone, vickers1996topology} for further discussion.

A \textbf{nucleus} is the algebraic structure on a frame that gives a sublocale. It given as a monotone function on a frame \(j : L \rightarrow L\), satisfying three properties. First: \(j(a \wedge b) = j(a) \wedge j(b)\). Second, \(a \le j(a)\). Finally, \(j(ja) \le j(a)\). These may be summarized as saying that it is (finite) meet-preserving, contractive (in this case, inflationary) and idempotent. An element of \(L\) is said to be \(j\)-closed if \(j(a)=a\). Further, if \(L\) is a frame, then \(L/j\), which consists only of \(j\)-closed elements, is also a frame, and there exists a surjective frame homomorphism \(j^* : L \rightarrow L/j\). If we view a frame as a category, a nucleus is just a left-exact idempotent monad; if we view a frame as a decategorified topos, a nucleus is a topology; and if we view a frame as generated by its internal logic, a nucleus gives a ``possibility'' or ``locally true'' modal operator on that logic, analogous to the S4 diamond.

A site of a topos is typically given as a category and an associated Grothendieck topology. But such topologies are somewhat painful to manipulate and reason about. In the special case where we are concerned only with the ``(0,1)-sheaves'' (i.e. truth-valued sheaves, i.e. downsets) of a poset, we can equally well just work with a site as a poset and a nucleus on its downsets.

\begin{definition}
The \textbf{Bruns-Lakser topology} on a finite locale is a function generated by the following stepwise procedure. First we consider the join-irreducible elements of the underlying poset of the locale (i.e. the poset of its join-irreducible elements). We term such elements the ``double basis'' of the locale (which consists of some, but not necessarily all, join-irreducible elements of the locale itself). For double-basis elements, we define \(bl\) as the identity. For all joins of such elements, \(bl\) is again the identity. And for all meets of all elements thus far enumerated, \(bl\) is again the identity. For all other elements, \(bl\) is the meet of all idempotent elements greater than or equal to it. Since we have explicitly added all necessary meets, the result is necessarily an idempotent.
\end{definition}

\begin{lemma}
The poset endomorphism, \(bl\), generated by the join-meet completion of the join-irreducible elements of the basis of a distributive lattice is a nucleus.
\end{lemma}
\begin{proof}
Since \(bl\) is idempotent by construction and obviously contractive, we only need show it preserves meets; i.e. that that \(bl(x) \wedge bl(y) = bl(x  \wedge  y)\).  If the meet is a meet of idempotents, then it is preserved by construction. If it is not, then it is the meet of elements which themselves are the meets of idempotents. As such, it may be rewritten to be a meet of idempotents, and thus is also preserved.
\end{proof}


\begin{lemma}
The Bruns-Lakser topology is the least topology which preserves the join-irreducible elements of the underlying poset of a finite frame.
\end{lemma}
\begin{proof}
We consider a finite frame as given by its underlying poset, i.e. as \(\Oc(P)\). Clearly \(\Oc(P)/bl\) contains elements corresponding to all join-irreducible elements in \(P\), and they remain join-irreducible as no elements are added below them. Further, as it is a frame, it is a distributive lattice. It remains to show that it is the least such lattice. The only idempotents which could introduce new join-irreducible elements in the induced basis are the meets of existing elements, of the form  \(z = x \wedge y\). Then \(z = (a \vee b) \wedge (c \vee d)\) for some four elements (possibly overlapping) which are themselves join-irreducible or joins of join-irreducible elements.  By the laws of distributive lattices we may rewrite between disjunctive and conjunctive normal forms and so \(z = \bigvee \{a \wedge c, a \wedge d, b \wedge c, b \wedge d\}\).  But now it is a join of smaller meets. Inductively, either these meets are join-irreducible themselves, or they may be further iteratively decomposed until they are, and therefore the original element is a join of irreducible elements.
\end{proof}

We can now state the central theorem of this section:

\begin{theorem}
The set of DSCs is one-to-one with the set of finite posites generated by downsets of upper semimodular lattices without \(M_3\), and equipped with the Bruns-Lakser topology.
\end{theorem}
\begin{proof}
The bulk of the work is handled by \ref{representation}, which establishes the connection between DSCs and upper semimodular lattices without \(M_3\). The correspondence between these and their downsets is established by Birkhoff duality, and the construction of a locale by Bruns-Lakser by the above lemma.
\end{proof}

\section{Version Parameterizations of Dependency Structures and Nuclei}

The mapping from an actual repository to a DSC is incomplete in that it throws away information about two packages being different ``versions'' of the same thing. Here we introduce a data structure that captures that additional information. A higher version of a package is, intuitively, something that is ``almost the same, but better''. Similarly, in a DSC, some events are in a sense ``universally more powerful'' than other events, in that they enable at least all the things enabled by the events they cover. We say an event is a ``higher version'' of another event if for possible dependency set of every event containing the lower event, there is also a possible dependency set for that event which differs only in that it contains the higher event instead.

\begin{definition}
A \textbf{version parameterization} of a DSC is an idempotent endofunction on events, from lower to higher, satisfying the above criteria; i.e. for every possible dependency set of every event, there is another possible dependency set of that event where the lower versions have been substituted for higher versions. Idempotency here translates into the condition that no higher event is itself a lower event of something else.
\end{definition}

Versioning parameterizations on DSCs in turn give rise to related structures on their reachable dependency posets. Since a reachable dependency poset is complete under joins, for every two event sets differing only in versions of some events, we can take the element corresponding to the union of their event sets. This yields a  endofunction on elements of the reachable dependency poset, sending (``contracting'') the lower and higher event sets both to their corresponding union. We term such an endofunction a poset version parameterization.

\begin{definition}
A \textbf{ponucleus} is an idempotent monotone endofunction on a poset which preserves existing finite meets.
\end{definition}

\begin{lemma}
Every poset version parameterization is a ponucleus that in addition preserves existing joins.
\end{lemma}

\begin{proof}
Clearly this function is idempotent and monotone by construction. It remains to show it preserves meets and joins.

First we consider meets. If two sets had a meet before, and one is now ``contracted'' to a higher version, this may introduce new elements in the intersection only on the condition that these elements were also in the second set. But that would mean that the elements of the second set would also be contracted in the same way, as would the elements in the intersection itself.

Next, we consider joins. If two sets had a join before, and one was contracted to a higher version, then this would introduce new elements in the union. But those elements are the same element which are introduced by contracting the union directly.
\end{proof}

Since poset version parameterizations preserve existing meets and joins, they certainly preserve meets and maximal joins, and thus are acted on by \(\BLc\).

\begin{theorem}
Given a join-preserving ponucleus \(j\), on a poset \(P\),  \(\BLc(j)\) is a nucleus on \(\BLc(P)\)
\end{theorem}
\begin{proof}
First we consider meet preservation. Since we know that \(j\) preserves meets, we only need observe that the action of \(\BLc\) on morphisms preserves meet-preservation. This follows from the fact that \(\BLc\) itself preserves meets.


Next, we consider contractivity. We already defined \(j\) to be contractive. So now we only need observe that \(\BLc\) preserves contractivity. For this to fail, it would need to extend a morphism so that some element was mapped to something below itself. But since contractivity is preserved for the basis join-irreducible elements, it must be preserved for all elements.

Finally, we consider idempotence. Again, this is given by idempotence of \(j\) combined with \(\BLc\) preserving idempotence of join-irreducible elements.
\end{proof}

As a corollary of the above, given a DSC \((E,D)\) and a version parameterization with an induced poset version parameterization \(p\), then \(\BLc(p)\) is a nucleus on \(\BLc(E,D)\).

To make this concrete, we consider our running example with three packages, such that \(a\) depends on \(b\) or \(c\), but now consider \(c\) to be a higher version of \(b\). The \(j\)-closed elements (fixpoints) of the induced nucleus on the \(rdp\) of this structure are then \(\emptyset\), \(bc\), and \(abc\), and the fixpoints of the induced nucleus on the Bruns-Lakser completion of such are \(\emptyset\), \(bc\), and \(a_ba_cbc\). This is illustrated in Fig. \ref{Fig4}. Note that every element has a unique least element of the set of fixed-points that is greater than or equal to it.

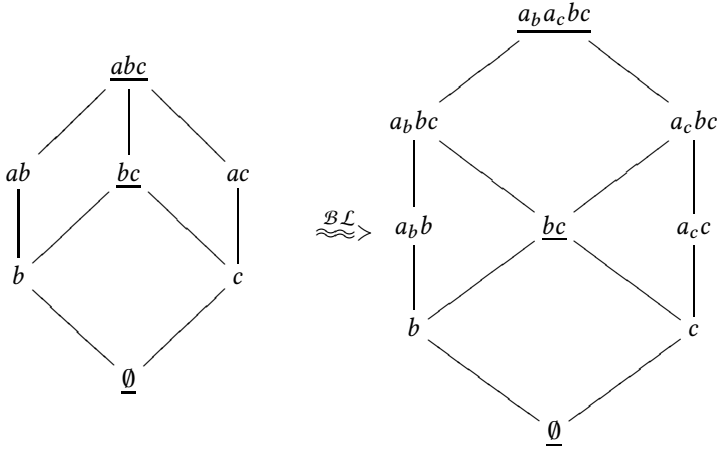
\begin{figure}
\begin{minipage}[c]{0.3\textwidth}
\begin{equation*}
    \xymatrix{& \underline{abc} \ar@{-}[dl] \ar@{-}[d] \ar@{-}[dr] & \\
      ab \ar@{-}[d] & \underline{bc} \ar@{-}[dl] \ar@{-}[dr] &
        ac \ar@{-}[d]  \\
        b \ar@{-}[dr] & & c \ar@{-}[dl]         \\
       & \underline{\emptyset} &}
\end{equation*}
\end{minipage}
\begin{minipage}[c]{0.08\textwidth}
  \begin{equation*}
    \xymatrix{\ar@2{~>}[r]^{\BLc} &}
  \end{equation*}
\end{minipage}
\begin{minipage}[c]{0.3\textwidth}
\begin{equation*}
    \xymatrix{& \underline{a_ba_cbc} \ar@{-}[dl]  \ar@{-}[dr] & \\
      a_bbc \ar@{-}[d] \ar@{-}[dr] & &
        a_cbc \ar@{-}[d] \ar@{-}[dl]\\
        a_bb \ar@{-}[d] & \underline{bc} \ar@{-}[dl] \ar@{-}[dr] &
          a_cc \ar@{-}[d] \\
         b \ar@{-}[dr] & & c \ar@{-}[dl] \\
      & \underline{\emptyset} }
\end{equation*}
\end{minipage}
\caption{The running example of choice, when equipped with a version policy that gives \(c\) as a higher version of \(b\), both on the reachable dependency poset, and its Bruns-Lakser completion. Fixed points of the induced nucleus are indicated by underline.}
\label{Fig4}
\end{figure}

\section{Free Distributive Lattices and the Modal Logic(s) of Dependencies}
As discussed above, part of the basic theory of Heyting algebras is that they possess an internal intuitionistic logic. Here we sketch how it works in our particular case.

Given a DSC \((E,D)\), we construct a language consisting of atoms given by join-irreducible elements of the reachable dependency poset (which may be thought of, as above, as events subscripted with their dependency trace), completed under the standard logical connectives. Every formula in this language corresponds to a particular element in the Merkle-lattice of our DSC, \(\BLc(E,D)\).  Conjunction corresponds to meet, disjunction to join, and implication to the relative pseudo-complement. This is a logic of reachable states of our system, and their traces, which describes all possible states of the system as disjunctions of join-irreducible states. Given two event sets, considered as reachable states, disjunction gives the set of events that have occurred in either state. Conjunction gives the set of events which have occurred in both states.


The pvp-induced nucleus discussed above equips the internal logic of \(\BLc(E,D)\) with a modal operator. We can interpret this operator as ``round (or upgrade) to the highest version'', and give corresponding interpretations of the meaning of its basic laws. \(x \rightarrow \Dia{x}\) tells us that everywhere \(x\) is valid, so too is its highest version (which is precisely how we constructed our modality to begin with). \(\Dia\Dia{x} \rightarrow \Dia{x}\) tells us that the highest version of the highest version is just the highest version. \((x \rightarrow y) \rightarrow (\Dia{x} \rightarrow \Dia{y})\)  tells us that if an implication holds for an event set, then the highest version of that set implies the highest version of the consequent. Finally, \(\Dia(x \band y) = \Dia{x} \band \Dia{y}\) tells us that the highest version of a conjunction may be computed as the conjunction of the highest versions of its constituents.

This modality is powerful enough to yield an internal ``bind'' operator. We have a strength that gives \(x \band \Dia{y} \rightarrow \Dia(x \band y)\). Furthermore, we have an internal evaluation \(x \band (x \rightarrow y) \rightarrow y\). Together, they allow us to generate a ``bind'' of the form \(\Dia{x} \band (x \rightarrow \Dia{y}) \rightarrow \Dia{y}\), read as ``if we know that some version of x implies the highest version of y, then to imply y it suffices to consider the highest version of x''. Hence, we have a computational monad in the sense of Moggi \cite{moggi1991notions}, and rounding (or ``upgrading'') is an ``effect''. 



From the standpoint of nondeterministic concurrent semantics, all of this is perfectly reasonable. From two configurations (considered as nondeterministic traces), ``or'' gives us their combination, and ``and'' gives us the least trace which they share between them. In a sense, given any concurrent program, this construction provides the internal logic of its scheduler. In a sense, concurrent programs each give rise to ``little programming languages'' that specify particular paths within them.

Nonetheless, from the standpoint of dependencies, this is insufficient. It is a logic of states in the system, but it is not a logic \textit{about} states in the system -- i.e., requirements. A a logical ``or''  would express that we're asking for one or the other event set as such, unlike the disjunction given here. Similarly, a logical ``and'' would express that we are requiring all events in both sets (i.e. what disjunction actually provides). Without this missing logical ``or'', for example, we cannot give a formula that specifies ``any way to reach an event, regardless of the dependencies'', as in the logic thus far, \(a_b \bor a_c\) specifies \textit{both} ways, as opposed to \textit{either} way.

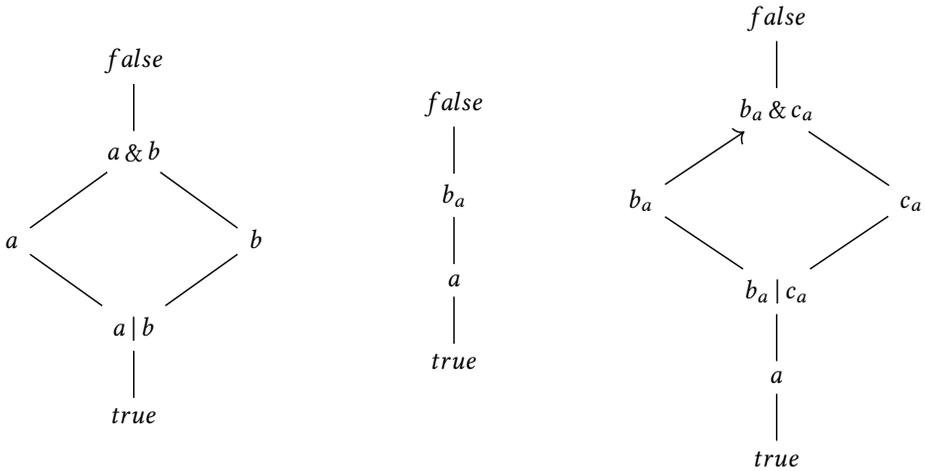
\begin{figure}
\begin{minipage}[c]{0.3\textwidth}
\begin{equation*}
\begin{tikzcd}
             & false                            &              \\
             & a \band b \arrow[u, no head]              &              \\
a \arrow[ru, no head] &                                 & b \arrow[lu, no head] \\
             & a \bor b \arrow[lu, no head] \arrow[ru, no head] &              \\
             & true \arrow[u, no head]                 &             
\end{tikzcd}
\end{equation*}
\end{minipage}
\begin{minipage}[c]{0.3\textwidth}
\begin{equation*}
\begin{tikzcd}
false                    \\
b_{a} \arrow[u, no head] \\
a \arrow[u, no head]     \\
true \arrow[u, no head] 
\end{tikzcd}
\end{equation*}
\end{minipage}
\begin{minipage}[c]{0.3\textwidth}
\begin{equation*}
\begin{tikzcd}
               & false                                                &                         \\
               & b_a \band c_a \arrow[u, no head]                     &                         \\
b_a \arrow[ru] &                                                      & c_a \arrow[lu, no head] \\
               & b_a \bor c_a \arrow[lu, no head] \arrow[ru, no head] &                         \\
               & a \arrow[u, no head]                                 &                         \\
               & true \arrow[u, no head]                              &                        
\end{tikzcd}
\end{equation*}
\end{minipage}
\caption{The free distributive lattice on two discrete elements,  on two elements with the first depending on the second. and on three elements with the latter two both depending on the first.}
\label{Fig5}
\end{figure}

Some further mathematical constructions are required to build back up to the logic we'd really like. Returning to the finite characterization of the extended Bruns-Lakser completion, we note that taking the downsets of the join-irreducible elements is effectively taking their free join-completion. Dually, taking up-sets is a free meet-completion. In fact, for a discrete set \(S\), \(\Ucc(\Oc(S))\) is the \textbf{free distributive lattice} over S, which is a well studied mathematical object \cite{gratzer2009lattice}. The elements of this belong to \(\Pc(\Pc(S))\) and consist of its \textit{irredundant subsets}. This is to say that we can read these elements as logical expressions in disjunctive normal form, for example as \(ab \bor c\). An irredundant set is one in which the clauses have been simplified -- i.e. in which \(a \bor ab\) has been reduced to simply \(a\).  In this construction, join remains union -- but now it is not a union of sets of atoms, but a union of \textbf{sets of sets} of atoms. Meet however, is no longer intersection. Rather, it becomes convolution, just as in standard logic! This is to say that the meet (and) of \(a \bor b\) and \(c \bor d\) becomes the irredundant (simplified) core of \(a \band c \bor a \band d \bor b \band c \bor b \band d\).

There are some interesting open problems regarding free distributive lattices, in particular the search for a closed form expression that counts the elements of such a construction over a set of a given size (aka the Dedekind numbers).

The free distributive lattice construction extends to any poset \(P\), where the up-sets of the downsets are the free distributive lattice of a poset. This is a less studied, but still known construction \cite{johnstone1982stone}. As above, join is union and meet is convolution. But when we ``multiply'' two atoms, we don't simply conjoin them. Rather, we take their join in the underlying poset. Hence we arrive at a quotient of the free distributive lattice generated when the elements of \(P\) are considered simply as a set, and potentially a much smaller one. For example, if the underlying poset is a linear order, then the free distributive lattice over it is equivalent to the original poset. In general, the more ordering in the original poset, the smaller the resultant free distributive lattice over it. See Fig. \ref{Fig5} for a few examples.

 In line with this, we can build the free meet-completion of the idempotent distributive lattice completion of a dependency poset by the compound construction \(\Ucc(\Oc(\Jc(P)))\). For a DSC, the internal logic of this construction, which we call \(\UBLc(E,D)\) has precisely the same syntax as that of the internal logic \(\BLc(E,D)\). However, while the normal form of the latter consisted of a set of irreducible elements, the normal form for our new construction consists of a set of sets -- i.e. a boolean expression of these elements given in disjunctive normal form, and simplified by reduction with regards to redundancy as well as the ordering relationship given on the join-irreducibles. We refer to the elements of \(\UBLc(E,D)\) not as event sets, but as event equations, since they consist of events related by conjunction and disjunction. Further, it should be noted that since \(\BLc\) is functorial, and taking up-sets is functorial, then the compound \(\UBLc\) is also functorial, and has an action on morphisms of posets.
 


To fully explore the dependency standpoint. we would also like to interpret the modal operator associated with a versioning parameterization in \(\UBLc(E,D)\) as well. This brings us to the following:

\begin{theorem}
Given a DSC \((E,D)\) and a version parameterization with an induced poset version parameterization \(p\), then \(\UBLc(p)\) is a nucleus on \(\UBLc(E,D)\). Furthermore, it preserves not only meets, but also joins.
\end{theorem}
\begin{proof}
We have already established that \(\BLc\) yields a contractive endofunction on a distributive lattice (and more). Since taking up-sets is functorial, we know that this in turn yields a contractive endo-function on the compound. Furthermore, since the up-set functor runs from all posets to distributive lattices, it lifts order-preserving morphisms of posets into meet and join preserving morphisms of distributive lattices. Hence, the induced contractive endofunction preserves both meets and joins, and is a nucleus.
\end{proof}

This in turn induces a modality, as above, but now on the logic not only of event sets, but of event equations. If we view reachable configurations as nondeterministic traces, then the ``highest version'' is most general state with regards to any given configuration -- i.e. associates to any trace another one which provides at least as many options for further execution as before, and possibly more. The various laws can be read similarly. 

Thus we get a simple type theory of package dependencies with an internal ``version policy'' modality that obeys the usual algebraic identities, and lets us treat contraction along the version policy as a monadic effect. Package version policies have been a subject of some debate in the applied world. But as we now see, all told, a package version policy in a dependency structure is just a certain type of monad in the free distributive lattice over the join-irreducibles of the associated poset.

\begin{figure}
\begin{equation*}
\begin{tikzcd}[row sep=scriptsize, column sep=large,cells={nodes={fill=white,fill opacity=1}}]
                            &  &                                                                          & false                                                                                              &                                                                          &  &                             \\
                            &  &                                                                          & \underline{a_b a_c} \arrow[u, no head]                                                                         &                                                                          &  &                             \\
a_b c \arrow[rrru, no head] &  &                                                                          &                                                                                                    &                                                                          &  & a_c b \arrow[lllu, no head] \\
a_b \arrow[u, no head]      &  &                                                                          & a_b c \bor a_c b \arrow[lllu, no head] \arrow[rrru, no head]                                       &                                                                          &  & a_c \arrow[u, no head]      \\
                            &  & a_b \bor a_c b \arrow[ru, no head] \arrow[rrrruu, no head]               & \underline{bc} \arrow[u, no head]                                                                              & a_c \bor a_b c \arrow[lu, no head] \arrow[lllluu, no head]               &  &                             \\[10pt]
                            &  & a_b \bor bc \arrow[ru, no head] \arrow[lluu, no head] \arrow[u, no head] & a_b \bor a_c \arrow[llluu, no head] \arrow[rrruu, no head] \arrow[lu, no head] \arrow[ru, no head] & bc \bor a_c \arrow[lu, no head] \arrow[rruu, no head] \arrow[u, no head] &  &                             \\
b \arrow[uuu, no head] \arrow[rrrdd, no head]   &  &                                                                          & a_b \bor bc \bor a_c \arrow[lu, no head] \arrow[u, no head] \arrow[ru, no head]                    &                                                                          &  & c \arrow[uuu, no head] \arrow[llldd, no head]               \\
                            &  & c \bor a_b \arrow[rrrru, no head] \arrow[ru, no head]                    &                                                                                                    & b \bor a_c \arrow[llllu, no head] \arrow[lu, no head]                    &  &                             \\[10pt]
                            &  &                                                                          & b \bor c \arrow[ru, no head] \arrow[lu, no head]     &                                                                          &  &                             \\
                            &  &                                                                          & true \arrow[u, no head]                                                                            &                                                                          &  &                            
\end{tikzcd}\end{equation*}
\caption{The free distributive lattice on the Bruns-Lakser completion of our running example, \(S_7\)}.
\label{Fig6}
\end{figure}
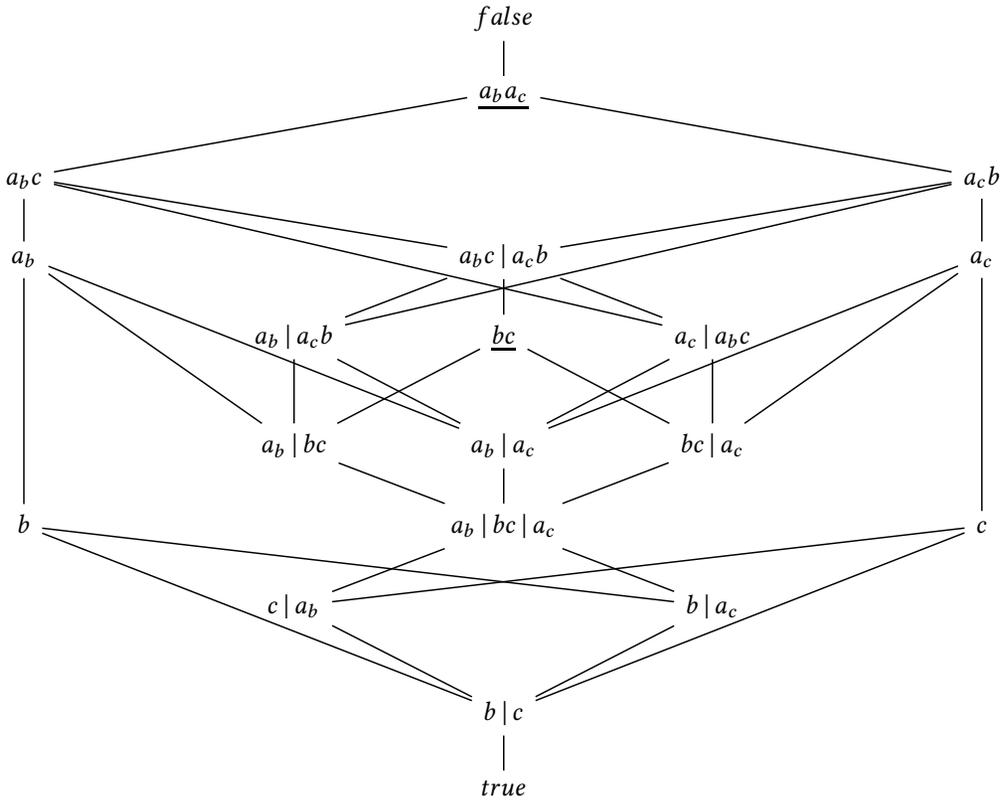

To illustrate this, we consider again our running example, with \(a\) depending on \(b\) or \(c\) and \(c\) a higher version of \(b\), and give the lattice of equations on it (Fig. \ref{Fig6}). Note that while this is quite large, it is considerably smaller than a full double powerset on packages (which, given four generators as \(a_b\) differs from \(a_c\), would have 256 elements), and smaller than the free distributive lattice on four discrete elements as well (which has 168 elements). So the dependency structure does help considerably in reducing the space under consideration. As before, elements which are fixed points of the induced version modality are underlined. If we consider only equations which are pvp-modal, then the size of the lattice is reduced dramatically further -- in this case the lattice of 21 elements reduces to a linear order of four elements (true, false, \(bc\), \(a_b a_c\)).

Clearly there are nuclei in the free distributive lattice over a poset which are not induced by nuclei on their join-irreducible elements, as not all nuclei preserve joins. Thus, we can imagine many other types of modal contraction operators which do not come from versioning relations, but which also reduce ``fuzzy'' specifications to more definite (unique) specifications of event sets. In  a sense these would all specify ways to contract away general powersets of events to singleton powersets, and can be seen as different specifications of dependency solvers. This provides some support for the slogan ``solving is an effect''.

\section{Solutions to Dependency Problems, and their Combinatorial Properties}
The above illustrates that when we wish to consider search problems regarding packages (or events), considering their dependency structures and their versioning structures can drastically reduce the complexity of the task. To make this precise, we define a general notion of a dependency problem, and show that the complexity of solving such problems may be related to order-theoretic properties of the dependency structures they are calculated over. On particular application of this is in finding solutions to package requirements that are free of incompatibilities.

\begin{definition}
A dependency problem in a DSC \((E,D)\) is the pair of a formula \(\phi\) in \(\UBLc(E,D)\) and a monotone increasing (i.e. growing as further elements are added to the source set) objective function of type \(\Pc(E) \rightarrow \mathbb{R}\). A solution to such a problem is an event set which satisfies the formula and minimizes the objective function.
\end{definition}

This naturally encodes many problems. For example, it allows us to calculate the minimal dependencies needed to reach a certain state. Further, if we associate a cost function which counts the number of incompatibilities or conflicts in a given event set, then a conflict-free solution is possible only when a minimal solution has a cost of zero.

Even without further constraints on the objective function, solving a dependency problem does not mean undertaking a brute force search over all event sets which satisfy the formula. In particular, since the objective function is monotone, we need only examine the minimal event sets which satisfy it (in the ordering induced by \(\BLc(E,D)\). Such sets in \(\BLc(E,D)\) form a maximal  antichain (or cut) -- i.e. they are unordered with relation to one another, but every other point on the lattice is ordered with regards to at least one such set. The former property follows immediately from monotonicity. The latter comes from the fact that the top element in the lattice satisfies all formulae in our language (since it has no notion of negation), and hence every set which does not satisfy \(\phi\) is ordered in relation to at least one set which does.  More precisely, we have the following lemma, which is straightforward, but we have not seen recorded in existing literature:

\begin{lemma}
The collection of all antichains in the downsets of a finite poset \(P\), under the inclusion ordering, corresponds to the free distributive lattice over \(P\).
\end{lemma}
\begin{proof}
We have already defined the free distributive lattice as the up-sets of the downsets. Hence we need only show that up-sets are one-to-one with antichains. In a finite lattice, every up-set corresponds to a unique set of basis elements, which are unordered with respect to one another, and hence form an antichain. Conversely, every antichain thus generates a distinct up-set.
\end{proof}

Since cuts are one-to-one with downsets (or equivalently up-sets), cuts correspond to elements of \(\UBLc(E,D)\) -- i.e. they are one-to-one with formulae in normal form.  

This encourages us to focus on the efficiency of enumerating maximal antichains, and in particular the maximum size of an antichain within a poset -- known as the width of the poset, which we denote as \(\Wf(P)\). This tells us, given a dependency structure, the maximum number of sets we have to consider in any dependency problem, with any specified formula. This is to say that the maximum width of the collection of downsets over a poset is the same as the maximum number of disjuncts in a formula over the join-irreducible elements of the poset presented in disjunctive normal form, and modulo appropriate relations.  

By Sperner's Theorem, first published in 1928, the powerset of a set with \(n\) elements, under inclusion ordering, has a width of \(\binom{n}{\ceil{n/2}}\).\footnote{Because this picks out the central elements of Pascal's triangle, using ceiling or floor here yields the same sequence. For symmetry with our more general result, contrary to standard convention, we prefer to use ceiling.} In the other extreme, when all elements of the underlying poset are strictly ordered, then all elements of the resultant downset are also strictly ordered, and hence the width is 1. This yields a key insight -- the complexity of solving a dependency problem is not strictly a function of the quantity of events -- rather, it is jointly determined by the quantity of events and the degree of dependency induced by their underlying topological structure. Sperner's Theorem lies at the foundation of a field of combinatorics known as extremal set theory and is closely related to the methods and tools of algebraic combinatorics. Using techniques from these fields, we can give bounds that capture this general relationship between width, size, and dependency degree.

A key step in doing so is a lemma provided to us by Richard Stanley \cite{343183}.

\begin{lemma}
\textbf{Stanley's Width Lemma (2019)}: Define on an integer \(h\), \((h) = 1 + x + x^2 \ldots x^{h-1}\). In a product of chains (linear orders considered as posets) of sizes \(h_1 \ldots h_n\), the width is given as the middle coefficient of the polynomial \((h_1) * (h_2) * \ldots * (h_n)\).
\end{lemma}

A corollary of this, which is more straightforward to compute is the following:  Define \(\Mf(a,b)\) as the central (maximal) coefficient of the formal polynomial expansion of \((1 + x + x^2 ... + x^a)^b\). (When \(a\) is 2, this is the central binomial coefficient, as appearing in Sperner's Theorem, as \(a\) increases this results in central coefficients of higher multinomials). Given a product of \(x\) chains all of size \(y\), then the width is \(\Mf(x,y)\).

With this in hand, we can provide an upper bound on the width of our dependency structures under consideration:

\begin{theorem}
Define \(\Hf(P)\) as the height of a poset, i.e. the length of its longest chain. Given any poset \(P\), then \(\Wf(\Oc(P)) \le \Mf(2,\Wf(P)) * \Mf(\Hf(P),\ceil{\Wf(P)/2})\).
\end{theorem}

\begin{proof} 
We first partition \(P\) into a collection of maximal antichains. There are, by definition \(\Wf(P)\) many of these. We next ``round'' each antichain up to \(\Hf(P)\). Now, by Sperner, the powerset of this collection of chains (considered as a discrete set) has the width given in the first half of our product. But we are interested not only in the powerset, but the ``power-product''. This is to say, for each choice of a collection of chains, we must \textit{also} make a choice of an element within each chain. Therefore, for each maximal subset of chains (which by Sperner necessarily have size \(\Wf(P)/2\)), we calculate the maximal size of an antichain within it. For that, we make use of Stanley's width lemma. The resultant product then gives the total maximal size of an antichain.
\end{proof}

When a poset is discrete (i.e., a set) then this collapses to a statement of the central inequality of Sperner's Theorem.

Clearly, as the width of a poset decreases, so too does the width of its downset lattice. This validates the observation that additional dependency structure on a collection of events (resp packages) allows a much more efficient traversal of a much smaller search space.

\section{Related and Future Work}

We hope that this paper has made a convincing case for the pervasiveness of dependency structures, and thus the range of possible future applications stemming from their study. This section is therefore rather broad-ranging due to the many avenues we believe future research may take.

Some key work in studying the semantics and solvability of package dependencies, including basic results about complexity (in particular NP-completeness),  has been conducted in \cite{di2006edos, abate2012dependency}. Dependency structures also play a heavy, although largely implict role in the formal treatment of build systems as studied in \cite{mitchell:shake_24_sep_2018}. A programming language approach to the semantics of version control has been studied in \cite{swierstra2014semantics}. Version control and attendant topological models have been studied in \cite{angiuli2014homotopical} and \cite{mimram2013categorical}. We believe our work makes a contribution by highlighting the essential similarities involved in these approaches, and its relationship to concurrent semantics. As one example of transfer in the other direction, we have provided a meaningful interpretation of the versioning parameterization modality in terms of dependency structures and packages. However, the same modality also may be interpreted in terms of concurrent semantics. What would it mean in such a case? One suggestion, which we have not yet explored, is that it may correspond to some form of modality for temporal logic, with different modes capturing ``eventually'', ``always'', ``soon'', and so forth.

Additionally, we believe that the general approach of modeling dependency structures may be of use in mechanized representations of knowledge (especially mathematical knowledge), such as that pursed by the Formal Abstracts project \cite{fabstract}. Also related to knowledge representation, we believe the notion of a Merkle-trace may be of some use in considering issues of logical atomism à la Wittgenstein.

As discussed in the introduction, the closest related work to this is the study of event structures. An important line of work on event structures has been seeking to recover the generality of a general event structure (with choice) while maintaining some form of representation theorem in domains, as in \cite{nielsen1981petri}. One recent article in this direction (\cite{DBLP:journals/corr/abs-1802-03726}), summarizes its key idea as: ``if in a general event structure an event has conflicting classes of causal histories, then it should split in several copies when generating the corresponding connected event structure.'' This insight is very close to the action taken by the Bruns-Lakser completion in the current work.

Our work also hopes to provide a contribution in the direction of reconciling choice and conflict in semantics of  general event structures, although conflict is not yet integrated into what is presented in the current paper. A central element of future work is to try to re-equip DSCs with a suitable notion of conflict. One key preliminary observation in that direction begins with the fact that conflict structures themselves (often known in event literature often as ``contexts''), are represented as a powerset of events that is downward-closed. That is, if some set \(s\) is a valid context (i.e. is conflict-free), then so too is every subset. The observation is that this same mathematical structure already has a direct topological interpretation -- a presentation of a simplicial complex (such that sets of cardinality \(n\) form the \(n\)-cells). The downward closure condition then amounts to the fact that to have a face, one must necessarily have its boundaries. Our hope is that from this observation, we can then find a way to relate the topology induced by dependency and choice to that induced by conflict into a unitary spatial structure.

It is our suspicion that focus on compositional semantics of concurrent processes has \textit{necessarily} led to simplifying assumptions / enforced behaviors in process calculi, because trace semantics of concurrent programs do not compose, since they are only the 1-cells of a more complicated topological object. If this is the case, then what composes is the full cellular structure, including the higher structure of incompatibilities, 2-incompatibilities, etc.

Another family of topological models of concurrency arises from directed algebraic topology, as in \cite{fajstrup2016directed}. In future work, we hope that the constructions here may be seen in such a light. In particular, many models of directed topology (in particular, those without self-loops) begin with a classical notion of a topological space constructed by points and data concerning open sets, and further equip those opens with some form of ordering. In a sense, this is again trying to reconcile two different innate topologies (those induced by the opens, corresponding to conflicts, and those induced by the orderings, corresponding to dependencies with choice). Topological information is then extracted from such structures by developing a notion of algebraic topology (homology, homotopy, etc.) over such directed spaces. Here, rather than having a full collection of opens, we hope to start with the localic ``skeleton'' of such, and study its homology relative to the space of conflict-free event sets. To do so requires picking (or, likely, constructing) a correct notion of finite homology suited to our purposes. We also suspect that methods of discrete Morse theory, as in \cite{forman2002user} may come into play, in particular when studying conflicts through the lens of dependency problems, where the conditions on such constructions seem to induce a discrete Morse structure.

The notion of equipping a spatial structure with a further poset structure goes back before
directed topology to the work of physicist Rafel Sorkin who suggested it as a way of approximating the
topological and causal structure of spacetime events \cite{sorkin1983posets}.  The constructions here and elsewhere in directed
topology may find use in helping to realize Sorkin’s original program (which later shifted into the study of topology-less ``causal sets'').  Another application to the natural sciences is to molecular biology where, in order to produce some metabolite, an organism needs to be able to produce
all precursors which are not present in the environment and one would like to understand how this requirement
constrains the possible structure of metabolic reaction networks.

The connection explored between lattices of certain forms and counting problems is one of a number of potential connections to the combinatorics of extremal set theory as in \cite{stanley2013algebraic}. Also potentially related is the Stanley-Reisner ring \cite{francisco2014survey}, which also connects lattice combinatorics to algebraic topology.




%





\begin{acks}
The initial impetus for this research grew out of many interesting conversations with Herbert Valerio Riedel regarding package repository semantics. The need to render these ideas more formal was due to the careful interrogation of Simon Peyton-Jones. Jonas Frey first pointed out to us the relationship to event structures. As this work has developed, it has relied on the constant input of participants in the NY Category Theory seminar, especially Thomas Lawler, James Deikun (who is very good at telling us what's wrong), and Jay Sulzberger (who is unmatched at helping us see what's right). We have also gained insight from useful discussions with Sridhar Ramesh and Martin Bendersky, as well as participants at Category Theory Octoberfest 2019 (especially Todd Trimble), and the MIT Categories Seminar (especially David Spivak and Brendan Fong), where early versions of some of this were presented. We are also very happy that Richard Stanley provided a gracious answer to our question on mathoverflow.
\end{acks}

\newpage
 \nocite{*}
 \bibliography{topo-logic-draft}

\end{document}